\documentclass[aps, pra, reprint, superscriptaddress, onecolumn, notitlepage, tightenlines, 11pt]{revtex4-2}

\pdfoutput=1

\usepackage{header}

\graphicspath{{graphics/}}

\begin{document}

\bibliographystyle{apsrev4-2}

\title{Quantum Parameterized Complexity}

\author{Michael J. Bremner}
\email{michael.bremner@uts.edu.au}
\affiliation{Centre for Quantum Computation and Communication Technology}
\affiliation{Centre for Quantum Software and Information}
\affiliation{School of Computer Science, Faculty of Engineering \& Information Technology, University of Technology Sydney, NSW 2007, Australia}

\author{Zhengfeng Ji}
\email{jizhengfeng@tsinghua.edu.cn}
\affiliation{Department of Computer Science and Technology, Tsinghua University, Beijing, China}
\affiliation{Centre for Quantum Software and Information}
\affiliation{School of Computer Science, Faculty of Engineering \& Information Technology, University of Technology Sydney, NSW 2007, Australia}

\author{Ryan L. Mann}
\email{mail@ryanmann.org}
\homepage{http://www.ryanmann.org}
\affiliation{Centre for Quantum Computation and Communication Technology}
\affiliation{Centre for Quantum Software and Information}
\affiliation{School of Computer Science, Faculty of Engineering \& Information Technology, University of Technology Sydney, NSW 2007, Australia}
\affiliation{School of Mathematics, University of Bristol, Bristol, BS8 1UG, United Kingdom}

\author{Luke Mathieson}
\email{luke.mathieson@uts.edu.au}
\affiliation{School of Computer Science, Faculty of Engineering \& Information Technology, University of Technology Sydney, NSW 2007, Australia}

\author{Mauro E.S. Morales}
\email{mauricio.moralessoler@student.uts.edu.au}
\affiliation{Centre for Quantum Computation and Communication Technology}
\affiliation{Centre for Quantum Software and Information}
\affiliation{School of Computer Science, Faculty of Engineering \& Information Technology, University of Technology Sydney, NSW 2007, Australia}

\author{Alexis T.E. Shaw}
\email{alexis@alexisshaw.com}
\affiliation{Centre for Quantum Computation and Communication Technology}
\affiliation{Centre for Quantum Software and Information}
\affiliation{School of Computer Science, Faculty of Engineering \& Information Technology, University of Technology Sydney, NSW 2007, Australia}

\begin{abstract}
    Parameterized complexity theory was developed in the 1990s to enrich the complexity-theoretic analysis of problems that depend on a range of parameters. In this paper we establish a quantum equivalent of classical parameterized complexity theory, motivated by the need for new tools for the classifications of the complexity of real-world problems. We introduce the quantum analogues of a range of parameterized complexity classes and examine the relationship between these classes, their classical counterparts, and well-studied problems. This framework exposes a rich classification of the complexity of parameterized versions of QMA-hard problems, demonstrating, for example, a clear separation between the Quantum Circuit Satisfiability problem and the Local Hamiltonian problem.
\end{abstract}

\maketitle

{
\hypersetup{linkcolor=black}
\tableofcontents
}

\section{Introduction}
\label{section:Introduction}

Quantum and classical complexity theory provide an essential tool for establishing potential benchmarks by which we can classify the cost of computing a given problem. It allows us to group applications by complexity classes that determine the asymptotic tractability of a problem. However, standard complexity-theoretic classification only considers a single parameter – usually the size of the input to the problem. In practice there may be considerably more structure to a problem than can be quantified by a single variable. This limitation can lead to characterisations that are too broad, such as grouping problems as typical for notoriously difficult complexity classes, e.g., \NP (non-deterministic polynomial time), \QMA (quantum Merlin Arthur), and \QCMA (quantum classical Merlin Arthur), and subsequently might suggest problems are intractable for a given set of parameters, when in realistic scenarios these problems could be practically solvable. For example the $k$-\textsc{Colouring} problem is trivial on trees and $3$-\textsc{SAT} is polynomial-time solvable when the number of variables or the number of clauses is fixed. By varying the parameters considered, we can gain a more nuanced understanding of the conditions that contribute to the tractability of problems.

Parameterized complexity theory was developed in the 1990s to enrich the complexity-theoretic analysis of problems that depend on a range of parameters~\cite{downey1999parameterized, flum2006parameterized, downey2013fundamentals, cygan2015parameterized}. The key motivation was to develop tools that more closely mirror the real-world use cases of heuristics. Central to the theory is an examination of how the complexity of a problem might change if a certain parameter $k$ varies independently of the instance size $n$. The parameterized analogue of the complexity class \PT (polynomial time) is the class \FPT (fixed-parameter tractable). Informally, it is the class of problems that can be solved in time $f(k) \cdot n^{O(1)}$ for some computable function $f$. A classic example of this is \textsc{Vertex Cover}, which is \NP-complete in its usual form but has an $O(n)$ time algorithm if we consider constant-size covers~\cite{chen2010improved} and hence is in \FPT. Another important application of parameterized complexity is its use in determining the potential intractability of a problem. Parameterized intractability classes include \paraNP (parameterized non-deterministic polynomial time), \XP (slice-wise polynomial time), \WP, and the \W hierarchy (weft hierarchy).

Parameterized complexity has been applied to the classical simulation of quantum systems. In particular, fixed-parameter tractable algorithms have been established for simulating quantum systems when parameterized by the treewidth~\cite{markov2008simulating} and by the number of non-Clifford gates~\cite{bravyi2016trading, bravyi2016improved, bravyi2019simulation}. This was further extended to quantum Merlin Arthur problems when parameterized by the number of non-Clifford gates in the verification circuit~\cite{arunachalam2022parameterized}.

\subsection{Quantum Parameterized Complexity Framework}
\label{section:QuantumParameterizedComplexityFramework}

We establish a quantum equivalent of classical parameterized
complexity theory, motivated by the need to establish new tools for the complexity-theoretic classification of real-world problems. We introduce the quantum analogues of a range of parameterized complexity classes (see Table~\ref{table:ClassicalvsQuantumParameterizedComplexityClasses}), and examine the relationship between these classes, their classical counterparts, and well-studied problems (see Fig.~\ref{figure:ComplexityClassesProblemsRelations}). In some cases, the quantum generalisations and their relationships follow almost directly from the equivalent classical definitions and their relationships. However, in other cases, the quantum generalisations are much less straightforward, for example, in the case of the quantum generalisation of the \W hierarchy.

\begin{table}[ht]
\centering
\begin{tabularx}{\textwidth}{ccX}
    \hline
    \textbf{Classical Class} & \textbf{Quantum Class} & \textbf{Description} \\
    \hline
    \FPT & \FPQT & Fixed-Parameter Quantum Tractable. \\
    \hline
    \paraNP & \paraQMA & Parameterized Quantum Merlin Arthur. \\
    & \paraQCMA & Parameterized Quantum Classical Merlin Arthur. \\
    \hline
    \XP & \XQP & Slice-Wise Quantum Polynomial Time. \\
    \hline
    \WP & \QWP & FPQT reducible to \textsc{Weight-$k$ Quantum Circuit Satisfiability}. \\
    & \QCWP & FPQT reducible to \textsc{Classical Weight-$k$ Quantum Circuit Satisfiability}. \\
    \hline
    $\W[t]$ & $\QW[t]$ & FPQT reducible to \textsc{Weight-$k$ Weft-$t$ Depth-$d$ Quantum Circuit Satisfiability}. \\
    & $\QCW[t]$ & FPQT reducible to \textsc{Classical Weight-$k$ Weft-$t$ Depth-$d$ Quantum Circuit Satisfiability}. \\
    & $\QCWc[t]$ & FPQT reducible to \textsc{Hamming Weight-$k$ Weft-$t$ Depth-$d$ Quantum Circuit Satisfiability}. \\
    \hline
\end{tabularx}
\caption{Classical parameterized complexity classes and their quantum analogues.}
\label{table:ClassicalvsQuantumParameterizedComplexityClasses}
\end{table}

\begin{figure}[ht]
    \centering
    \includegraphics[width=\textwidth]{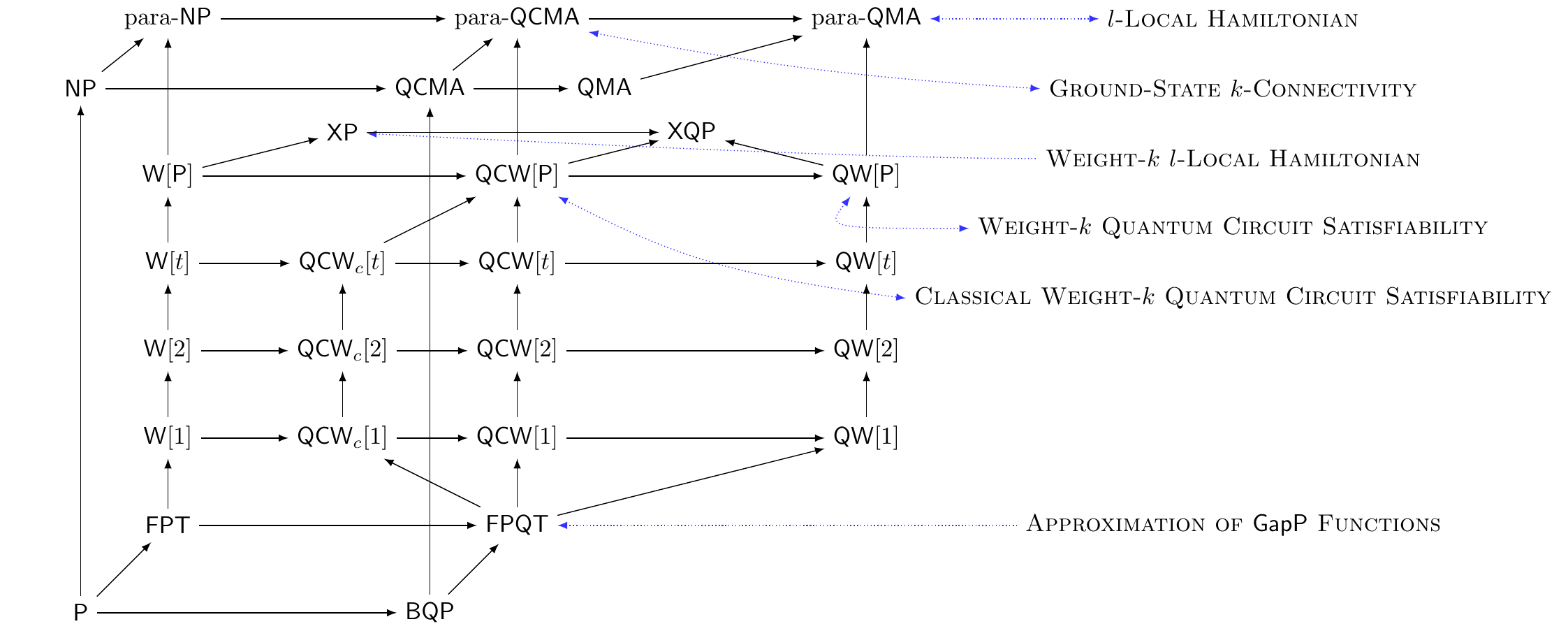}
    \caption{Complexity classes and problems discussed in this paper.}
    \label{figure:ComplexityClassesProblemsRelations}
\end{figure}

For quantum parameterized tractability, we introduce the complexity class \FPQT (fixed-parameter quantum tractable; Definition~\ref{definition:FixedParameterQuantumTractable}). Informally, \FPQT is the class of parameterized decision problems that are tractable on a quantum computer. We discuss the relationship of \FPQT with other well-studied complexity classes. While the generalisation of \FPT to \FPQT is relatively straightforward, the quantum generalisations of parameterized intractability classes is more complicated. As is standard in quantum complexity theory, we consider \QMA and \QCMA as the natural generalisations of the class \NP. This gives the parameterized generalisations of \paraNP as \paraQMA (parameterized quantum Merlin Arthur; Definition~\ref{definition:paraQMA}) and \paraQCMA (parameterized quantum classical Merlin Arthur; Definition~\ref{definition:paraQCMA}). The parameterized class \XP also readily generalises to the quantum case as \XQP (slice-wise quantum polynomial time; Definition~\ref{definition:XQP}). In classical parameterized complexity theory the \W hierarchy, its finite levels $\W[t]$, and their asymptotic limit \WP play an important role in bounding \FPT. We introduce quantum generalisations of the \W hierarchy via the \textsc{Weight-$k$ Quantum Circuit Satisfiability} problem to establish \QWP (Definition~\ref{definition:QWP}) and the \QW hierarchy (quantum weft hierarchy; Definition~\ref{definition:QWT}) and the \textsc{Classical Weight-$k$ Quantum Circuit Satisfiability} problem to establish \QCWP (quantum classical weft hierarchy; Definition~\ref{definition:QCWP}) and the \QCW hierarchy (Definition~\ref{definition:QCWT}). Further, we introduce a variation of the \QCW hierarchy via the \textsc{Hamming Weight-$k$ Quantum Circuit Satisfiability} problem to establish the \QCWc hierarchy (Definition~\ref{definition:QCWCT}).

We establish several structural results concerning quantum and classical parameterized complexity classes, for example, we show that $\FPT=\FPQT$ if and only if $\PT=\BQP$ (Proposition~\ref{proposition:FPTvsFPQT}). We also establish key technical components of quantum parameterized complexity such as FPQT reductions (Section~\ref{section:FixedParameterQuantumTractableReductions}). Further, we apply the notion of fixed-parameter quantum tractability to the problem of approximate counting (Section~\ref{section:ApproximateCounting}) and approximating quantum circuit probability amplitudes (Theorem~\ref{theorem:AdditiveErrorApproximationAlgorithmAmplitudes} and Corollary~\ref{corollary:MultiplicativeErrorApproximationAlgorithmAmplitudes}).

One of our most important observations concerns the complexity of weighted quantum Merlin Arthur problems, i.e., where the witness state is constrained to be a superposition of $n$-bit strings of Hamming weight $k$. We show that the \textsc{Weight-$k$ Quantum Circuit Satisfiability} problem is \QWP-complete under FPQT reductions (Proposition~\ref{proposition:WeightKQuantumCircuitSatisfiabilityQWPComplete}) and the \textsc{Weight-$k$ $l$-Local Hamiltonian} problem is in \XP (Proposition~\ref{proposition:WeightLocalHamiltonianinXP}). Since \textsc{Weight-$k$ Quantum Circuit Satisfiability} cannot be in \XP unless $\PT=\BQP$, this demonstrates a clear separation between the two problems.

\subsection{Discussion and Open Problems}
\label{section:Discussion and Open Problems}

There are several important open problems that still remain. The most compelling of these is the question of how to saturate the power of \FPQT in a natural way. While it is possible to construct problems for FPT-sized quantum circuits, this does not seem to naturally capture the role of the parameter. Further, it seems that many known \QMA-complete problems do not have parameterizations that are known to be in \FPQT. This is in contrast to the classical case, where there are several \NP-complete problems whose parameterizations are in \FPT.

Another important open problem is to identify natural complete problems for finite levels of the \QW hierarchy. While we establish natural complete problems for \paraQMA (\textsc{$l$-Local Hamiltonian}; Corollary~\ref{corollary:LLocalHamiltonianParaQMAComplete}) and \paraQCMA (\textsc{Ground State $k$-connectivity}; Corollary~\ref{corollary:GroundStateKConnectivityParaQCMAComplete}), we have been unable to establish natural complete problems for finite levels of the \QW hierarchy. This is contrast to the \W hierarchy, where there are several examples of natural complete problems, e.g., \textsc{$k$-Independent Set} for $\W[1]$ and \textsc{$k$-Dominating Set} for $\W[2]$~\cite{downey2013fundamentals}.

\section{Quantum Parameterized Tractability}
\label{section:QuantumParameterizedTractability}

We begin by introducing the theory encompassing tractability for parameterized problems. A \emph{parameterized language} is a language equipped with an additional specified input, the \emph{parameter}. Historically, parameterized problems have been defined using either an explicit parameter~\cite{downey1999parameterized, downey2013fundamentals}, or by defining a \emph{parameterization}~\cite{flum2006parameterized}. While there are subtle differences between these two approaches, for practical purposes they are equivalent, and we use the explicit parameterization.
\begin{definition}[Parameterization]
    A parameterization of a finite alphabet $\Sigma$ is a mapping $\kappa:\Sigma^*\to\mathbb{Z}^+$ that is polynomial-time computable. The trivial parameterization $\kappa_\text{trivial}$ is the parameterization with $\kappa_\text{trivial}(x)=1$ for all $x\in\Sigma^*$.
\end{definition}
  
We now define a \emph{parameterized problem}.
\begin{definition}[Parameterized problem]
    A parameterized problem over a finite alphabet $\Sigma$ is a pair $(L,\kappa)$ where $L\subseteq\Sigma^*$ is a set of strings over $\Sigma$ and $\kappa$ is a parameterization of $\Sigma$. We say that a parameterized problem $(L,\kappa)$ over the alphabet $\Sigma$ is \emph{trivial} if either $L=\varnothing$ or $L=\Sigma^*$.
\end{definition}

It is sometimes useful to consider parameterized problems for fixed values of the parameter. To do so the formal definition of a \emph{slice} of a parameterized problem is used.
\begin{definition}[Slice]
    Let $(L,\kappa)$ be a parameterized problem over the finite alphabet $\Sigma$ and let $l\in\mathbb{Z}^+$ be a positive integer. The $l^{\mathrm{th}}$ slice of $(L,\kappa)$ is the standard problem
    \begin{equation}
        (L,\kappa)_l \coloneqq \{x \in L \mid \kappa(x)=l\}. \notag
    \end{equation}
\end{definition}

We now describe the foundational definitions for tractability of quantum parameterized problems.

\subsection{Fixed-Parameter Quantum Tractable}
\label{section:FixedParameterQuantumTractable}

The central complexity class for establishing tractability in the quantum parameterized framework is Fixed-Parameter Quantum Tractable (\FPQT). Informally, an FPQT algorithm is a quantum algorithm that, for a parameterized problem $(L,\kappa)$, decides if $x$ is a member of $L$ with error probability at most $1/3$ in time $f(\kappa(x))\cdot\abs{x}^{O(1)}$ for some computable function $f$.
\begin{definition}[FPQT algorithm]
    Let $(L,\kappa)$ be a parameterized problem over the alphabet $\Sigma$. An algorithm $\mathcal{A}$ is a \emph{FPQT algorithm for $(L,\kappa)$} if the following conditions are satisfied.
    \begin{enumerate}
        \item There is a computable function $f:\mathbb{Z}^+\to\mathbb{Z}^+$ and a polynomial $p\in\mathbb{N}[X]$, such that, for every $x\in\Sigma^*$, the size of an FPT-uniform quantum circuit that computes $\mathcal{A}$ on input $x$ is at most $f(\kappa(x)) \cdot p(\abs{x})$.
        \item For every $x\in\Sigma^*$,
        \begin{itemize}
            \item If $x \in L$, then $\textbf{Pr}[\text{$\mathcal{A}(x)$ accepts}]\geq\frac{2}{3}$.
            \item If $x \notin L$, then $\textbf{Pr}[\text{$\mathcal{A}(x)$ accepts}]\leq\frac{1}{3}$.
        \end{itemize}
    \end{enumerate}
\end{definition}
The error probability of $1/3$ is completely arbitrary and can be replaced by any constant non-zero probability less than $1/2$. Note that we have adopted a uniform notion of an FPQT algorithm, however, it is also possible to adopt a non-uniform notion. We shall now introduce the complexity class \FPQT, which consists of all parameterized problems with an FPQT algorithm.
\begin{definition}[\FPQT]
    \label{definition:FixedParameterQuantumTractable}
    The class \FPQT consists of all parameterized problems that have an FPQT algorithm.
\end{definition}

\FPQT is the quantum analogue of the classical complexity class Fixed-Parameter Polynomial time (\FPT)~\cite{downey1999parameterized} and the parameterized version of the quantum complexity class Bounded-error Quantum Polynomial time (\BQP)~\cite{bernstein1997quantum}. It is easy to see that the slices of any problem in \FPQT are in \BQP.

\begin{proposition}
    Let $(L,\kappa)$ be a parameterized problem and let $l\in\mathbb{Z}^+$ be a positive integer. If $(L,\kappa)$ is in \FPQT, then $(L,\kappa)_l$ is in \BQP.
\end{proposition}
\begin{proof}
    The proof follows from the fact that $\kappa$ is polynomial-time computable.
\end{proof}

We shall now present some alternative characterisations of \FPQT.
\begin{theorem}
    \label{theorem:FPQTAlternativeCharacterisations}
    Let $(L,\kappa)$ be a parameterized problem over the alphabet $\Sigma$. Then the following statements are equivalent.
    \begin{enumerate}
        \item \label{theorem:FPQTAlternativeCharacterisations1} $(L,\kappa)$ is in \FPQT.
        \item \label{theorem:FPQTAlternativeCharacterisations2} $(L,\kappa)$ is in \BQP after a precomputation on the parameter. That is, there exists an alphabet $\Pi$, a computable function $\pi:\mathbb{Z}^+\to\Pi^*$, and a problem $X\subseteq\Sigma^*\times\Pi^*$ such that $X$ is in \BQP and, for all instances $x$ of $L$, we have $x \in L$ if and only if $(x,\pi(\kappa(x))) \in X$.
        \item \label{theorem:FPQTAlternativeCharacterisations3} $L$ is decidable and $(L,\kappa)$ is eventually in \BQP. That is, there exists a computable function $\rho:\mathbb{Z}^+\to\mathbb{Z}^+$ and a polynomial-time quantum algorithm that on input $x\in\Sigma^*$ with $\abs{x}\geq\rho(\kappa(x))$, decides if $x$ is a member of $L$ with error probability at most $1/3$.
    \end{enumerate}
\end{theorem}
\begin{proof}
    The proof follows similarly to that of the equivalent classical theorem~\cite[Theorem 1.37]{flum2006parameterized}.
\end{proof}

The following containment is straightforward.
\begin{proposition}
    $\FPT\subseteq\FPQT$.
\end{proposition}
We conjecture that this containment is strict, i.e., $\FPT\neq\FPQT$. However, proving a separation between \FPT and \FPQT is as difficult as proving a separation between \PT and \BQP. Therefore, resolving this conjecture is a hard open problem.
\begin{proposition}
    \label{proposition:FPTvsFPQT}
    $\FPT=\FPQT$ if and only if $\PT=\BQP$.
\end{proposition}
\begin{proof}
    Suppose that $\FPT=\FPQT$. For every problem $L\subseteq\Sigma^*$ in \BQP, we have that $(L,\kappa_\text{trivial})$ is in \FPQT. It then follows that $(L,\kappa_\text{trivial})$ is in \FPT and therefore $L$ is in \PT. Hence $\PT=\BQP$.

    Now suppose instead that $\PT=\BQP$, then the characterisation of \FPQT given by Theorem~\ref{theorem:FPQTAlternativeCharacterisations} (2.) is equivalent to a characterisation of \FPT~\cite[Theorem 1.37 (2.)]{flum2006parameterized}. Hence $\FPT=\FPQT$.
\end{proof}

We now prove some results that characterise \FPQT.
\begin{proposition}
    $\FPT^\FPQT=\FPQT$.
\end{proposition}
\begin{proof}
    The inclusions $\FPQT\subseteq\FPT^\FPQT$ and $\FPT^\FPQT\subseteq\FPQT^\FPQT$ are trivial. We have that $\FPQT^\FPQT=\FPQT$ by a similar proof to that showing $\BQP^\BQP=\BQP$~\cite{bernstein1997quantum}. Hence, $\FPT^\FPQT=\FPQT$, completing the proof.
\end{proof}

The next result gives an alternative characterisation for \FPQT in terms of $\FPT^\BQP$. As a consequence of this result, we can use \BQP-complete problems to define \FPQT-complete problems.

\begin{proposition}
    $\FPT^\BQP=\FPT^\FPQT$.
\end{proposition}
\begin{proof}
    The inclusion $\FPT^\BQP\subseteq\FPT^\FPQT$ is trivial. To show that $\FPT^\FPQT\subseteq\FPT^\BQP$, we apply Theorem~\ref{theorem:FPQTAlternativeCharacterisations}, which states that $(L,\kappa)\in\FPQT$ if and only if $L\in\BQP$ after a precomputation of the parameter. This precomputation can be performed by an \FPT machine, completing the proof.
\end{proof}

We present some application of fixed-parameter quantum tractability in Section~\ref{section:Applications}.

\subsection{Fixed-Parameter Quantum Tractable Reductions}
\label{section:FixedParameterQuantumTractableReductions}

We shall now introduce the notion of a reduction in quantum parameterized complexity that we use throughout the remainder of the paper.

\begin{definition}
   \label{definition:FPQTReduction}
    Let $(L,\kappa)$ and $(L',\kappa')$ be parameterized problems over the alphabets $\Sigma$ and $\Sigma'$ respectively. A \emph{FPQT reduction} from $(L,\kappa)$ to $(L',\kappa')$ is a mapping $R:\Sigma^*\to(\Sigma')^*$ such that the following conditions are satisfied.
    \begin{enumerate}
        \item For all $x\in\Sigma^*$, $x \in L \iff R(x) \in L'$.
        \item $R$ is computable by an FPQT algorithm with respect to the parameter $\kappa$, (i.e. $R(x)$ is computable using an \FPT-uniform collection of circuits of size $f(\kappa(x) \cdot p(|x|)$ with high probability).
        \item There is a computable function $g:\mathbb{Z}^+\to\mathbb{Z}^+$ such that $\kappa'(R(x)) \leq g(\kappa(x))$ for all $x\in\Sigma^*$.
    \end{enumerate}
\end{definition}

This definition gives the crucial property that \FPQT is closed under FPQT reductions.

\begin{proposition}
    \label{proposition:FPQTClosedUnderFPQTReductions}
    \FPQT is closed under FPQT reductions. That is, if $(L',\kappa')$ is in \FPQT and there is an FPQT reduction from $(L,\kappa)$ to $(L',\kappa')$, then $(L,\kappa)$ is in \FPQT.
\end{proposition}
\begin{proof}
    Let $(L',\kappa')$ be a parameterized problem in \FPQT and let $R$ be an FPQT reduction from $(L,\kappa)$ to $(L',\kappa')$ computable in time $g(\kappa(x)) \cdot q(\abs{x})$ with $\kappa(x') \leq g(\kappa(x))$, where $g$ and $h$ are computable functions and $q\in\mathbb{N}[X]$ is a polynomial. Let $\mathcal{A}$ be an FPQT algorithm for deciding $(L',\kappa')$ in time $f(\kappa(x)) \cdot p(\abs{x})$ with error probability at most $1/3$. Then we can decide if $x$ is a member of $L$ by firstly computing $R(x)$ and then deciding if $R(x)$ is a member of $L'$. This requires time at most $g(\kappa(x)) \cdot q(\abs{x})+f(\kappa(x)) \cdot p(g(\kappa(x))) \cdot p(q(\abs{x}))$. Then by applying a simple error gap argument, we obtain an FPQT algorithm for $(L,\kappa)$. Hence, $(L,\kappa)$ is in \FPQT.
\end{proof}

\section{Quantum Parameterized Intractability}
\label{section:QuantumParameterizedIntractability}

Classical parameterized complexity has a very rich theory of intractability with a series of fine-grained complexity hierarchies. This richness is reflected in the quantum case and perhaps even more so. In the following section we introduce several quantum analogues to the classical classes \paraNP, \XP, \WP and the \W hierarchy. In contrast to the apparently straightforward comparison between \FPT and \FPQT, the quantum intractability classes reveal interesting aspects particular to quantum computation.

\subsection{Parameterized Quantum Merlin Arthur}
\label{section:ParameterizedQuantumMerlinArthur}

The most immediate intractable classical parameterized class is \paraNP --- a direct parameterized analogue of \NP.
\begin{definition}[\paraNP]
    A parameterized problem $(L,\kappa)$ over the alphabet $\Sigma$ is in \paraNP if there is a verification procedure $\{\mathcal{V}_{n,k}\}_{n,k\in\mathbb{Z}^+}$ such that the following conditions are satisfied.
    \begin{enumerate}
        \item There is a computable function $f:\mathbb{Z}^+\to\mathbb{Z}^+$ and a polynomial $p\in\mathbb{N}[X]$, such that, for every $x\in\Sigma^*$, $\mathcal{V}_{\abs{x},\kappa(x)}$ on input $x$ runs in time at most $f(\kappa(x)) \cdot p(\abs{x})$ on a deterministic Turing machine.
        \item For every $x\in\Sigma^*$,
        \begin{itemize}
            \item If $x \in L$, then there exists a bit string $y$, such that $\mathcal{V}_{\abs{x},\kappa(x)}(x,y)$ accepts.
            \item If $x \notin L$, then for every bit string $y$, $\mathcal{V}_{\abs{x},\kappa(x)}(x,y)$ rejects.
        \end{itemize}
    \end{enumerate}
\end{definition}

The equivalent notion to \NP-completeness is obtained naturally using \FPT-reductions, however there is also an immediate theorem~\cite[Theorem 2.14]{flum2006parameterized} demonstrating the intractability of \paraNP-complete problems where problems that are \NP-complete for any finite set of values of the parameter are \paraNP-complete. Thus several standard \NP-complete problems are \paraNP-complete with their obvious parameterizations. For example, $k$-\textsc{Colouring} is \paraNP-complete when parameterized by the number of colours $k$, $k$-\textsc{SAT} is \paraNP-complete when parameterized by the size of the clauses $k$. We give the analogous theorem for \QMA-complete problems in Theorem~\ref{theorem:ParaQMACompleteness}.

In the context of quantum complexity classes, Quantum Merlin Arthur (\QMA)~\cite{watrous2000succinct} occupies a place congruent to \NP in the classical case. We introduce the class parameterized Quantum Merlin Arthur (\paraQMA), as the quantum analogue of \paraNP~\cite{flum2006parameterized} and the parameterized version of \QMA.
\begin{definition}[\paraQMA]
    \label{definition:paraQMA}
    A parameterized problem $(L,\kappa)$ over the alphabet $\Sigma$ is in \paraQMA{$(c,s)$} if there is a quantum verification procedure $\{\mathcal{V}_{n,k}\}_{n,k\in\mathbb{Z}^+}$ such that the following conditions are satisfied.
    \begin{enumerate}
        \item There is a computable function $f:\mathbb{Z}^+\to\mathbb{Z}^+$ and a polynomial $p\in\mathbb{N}[X]$, such that, for every $x\in\Sigma^*$, the size of an FPT-uniform quantum circuit that computes $\mathcal{V}_{\abs{x},\kappa(x)}$ on input $x$ is at most $f(\kappa(x)) \cdot p(\abs{x})$.
        \item For every $x\in\Sigma^*$,
        \begin{itemize}
            \item If $x \in L$, then there exists a quantum state $\ket{\psi}$, such that $\textbf{Pr}[\text{$\mathcal{V}_{\abs{x},\kappa(x)}(x,\ket{\psi})$ accepts}] \geq c$.
            \item If $x \notin L$, then for every quantum state $\ket{\psi}$, $\textbf{Pr}[\text{$\mathcal{V}_{\abs{x},\kappa(x)}(x,\ket{\psi})$ accepts}] \leq s$.
        \end{itemize}
    \end{enumerate}
    The class \paraQMA is defined to be \paraQMA{$(\frac{2}{3},\frac{1}{3})$}.
\end{definition}

We shall now present some alternative characterisations of \paraQMA.
\begin{proposition}
    \label{proposition:ParaQMAAlternativeCharacterisations}
    Let $(L,\kappa)$ be a parameterized problem over the alphabet $\Sigma$. Then the following statements are equivalent.
    \begin{enumerate}
        \item $(L,\kappa)$ is in \paraQMA.
        \item $(L,\kappa)$ is in \QMA after a precomputation on the parameter. That is, there exists an alphabet $\Pi$, a computable function $\pi:\mathbb{Z}^+\to\Pi^*$, and a problem $X\subseteq\Sigma^*\times\Pi^*$ such that $X$ is in \QMA and, for all instances $x$ of $L$, we have $x \in L$ if and only if $(x,\pi(\kappa(x))) \in X$.
        \item $L$ is decidable and $(L,\kappa)$ is eventually in \QMA. That is, there exists a computable function $\rho:\mathbb{Z}^+\to\mathbb{Z}^+$ and a QMA algorithm that on input $x\in\Sigma^*$ with $\abs{x}\geq\rho(\kappa(x))$, decides if $x$ is a member of $L$ with error probability at most $1/3$.
    \end{enumerate}
\end{proposition}
\begin{proof}
    The proof follows similarly to that of Theorem~\ref{theorem:FPQTAlternativeCharacterisations}.
\end{proof}

The following containments are straightforward.
\begin{proposition}
    $\paraNP\subseteq\paraQMA$ and $\FPQT\subseteq\paraQMA$.
\end{proposition}
We conjecture that these containments are strict, i.e., $\paraNP\neq\paraQMA$ and $\FPQT\neq\paraQMA$. However, proving a separation between \paraNP and \paraQMA is as difficult as proving a separation between \NP and \QMA, and proving a separation between \FPQT and \paraQMA is as difficult as proving a separation between \BQP and \QMA. Therefore, resolving these conjectures is a hard open problem.
\begin{proposition}
    \label{proposition:ParaNPvsParaQMA}
    $\paraNP=\paraQMA$ if and only if $\NP=\QMA$.
\end{proposition}
\begin{proof}
    Suppose that $\paraNP=\paraQMA$. For every problem $L\subseteq\Sigma^*$ in \QMA, we have that $(L,\kappa_\text{trivial})$ is in \paraQMA. It then follows that $(L,\kappa_\text{trivial})$ is in \paraNP and therefore $L$ is in \NP. Hence $\NP=\QMA$.

    Now suppose instead that $\NP=\QMA$, then the characterisation of \paraQMA given by Proposition~\ref{proposition:ParaQMAAlternativeCharacterisations} (2.) is equivalent to a characterisation of \paraNP~\cite[Proposition 2.12 (2.)]{flum2006parameterized}. Hence $\paraNP=\paraQMA$.
\end{proof}

\begin{proposition}
    \label{proposition:FPQTvsParaQMA}
    $\FPQT=\paraQMA$ if and only if $\BQP=\QMA$.
\end{proposition}
\begin{proof}
    Suppose that $\FPQT=\paraQMA$. For every problem $L\subseteq\Sigma^*$ in \QMA, we have that $(L,\kappa_\text{trivial})$ is in \paraQMA. It then follows that $(L,\kappa_\text{trivial})$ is in \FPQT and therefore $L$ is in \BQP. Hence $\BQP=\QMA$.

    Now suppose instead that $\BQP=\QMA$, then the characterisation of \FPQT given by Theorem~\ref{theorem:FPQTAlternativeCharacterisations} (2.) and the characterisation of \paraQMA given by Proposition~\ref{proposition:ParaQMAAlternativeCharacterisations} (2.) are equivalent. Hence $\FPQT=\paraQMA$.
\end{proof}

We shall now develop the theory of \paraQMA-completeness. Firstly, we show that \paraQMA is closed under FPQT reductions.
\begin{proposition}
    \paraQMA is closed under FPQT reductions.
\end{proposition}
\begin{proof}
    The proof follows similarly to that of Proposition~\ref{proposition:FPQTClosedUnderFPQTReductions}.
\end{proof}

The following theorem allows us to establish the \paraQMA-completeness of a wide range of problems.
\begin{theorem}
    \label{theorem:ParaQMACompleteness}
    Let $(L,\kappa)$ be a non-trivial parameterized problem in \paraQMA. Then the following statements are equivalent.
    \begin{enumerate}
        \item \label{ParaQMACompleteness1} $(L,\kappa)$ is \textnormal{\paraQMA-complete} under FPQT reductions.
        \item \label{ParaQMACompleteness2} The union of a finite number of slices of $(L,\kappa)$ is \textnormal{\QMA-complete}.
    \end{enumerate}
\end{theorem}
\begin{proof}
    The proof follows similarly to that of the equivalent classical theorem~\cite[Theorem 2.14]{flum2006parameterized}.
\end{proof}

We now proceed to show that the \textsc{$l$-Local Hamiltonian} problem is \paraQMA-complete.
\begin{center}
\begin{tabularx}{\linewidth}{l X c}
    \multicolumn{2}{l}{\textsc{$l$-Local Hamiltonian}:} \\
    \textit{Instance:} & An $l$-local Hamiltonian $H\coloneqq\sum_iH_i$ on $n$ qubits that comprises at most a polynomial in $n$ many terms $\{H_i\}$, which each act non-trivially on at most $l$ qubits and have operator norm $\norm{H_i}$ bounded from above by a polynomial in $n$. Two positive numbers $a,b\in(0,1)$, such that $b-a>\frac{1}{\text{poly}(n)}$. \\
    \textit{Parameter:} &  A natural number $l\geq2$. \\
    \textit{Problem:} & Decide whether $H$ has an eigenvalue less than or equal to $a$ or all of the eigenvalues of $H$ are greater than or equal to $b$, given the promise that one of these is the case.
\end{tabularx}
\end{center}

\begin{corollary}
    \label{corollary:LLocalHamiltonianParaQMAComplete}
    \textsc{$l$-Local Hamiltonian} is \textnormal{\paraQMA-complete}.
\end{corollary}
\begin{proof}
    The proof follows from Theorem~\ref{theorem:ParaQMACompleteness} and the fact that \textsc{$l$-Local Hamiltonian} is \QMA-complete for constant $l$~\cite{kitaev2002classical, kempe2006complexity}.
\end{proof}

\subsection{Parameterized Quantum Classical Merlin Arthur}
\label{section:ParameterizedQuantumClassicalMerlinArthur}

We shall now introduce the complexity class parameterized Quantum Classical Merlin Arthur (\paraQCMA), which is the subclass of \paraQMA restricted to classical proofs. Alternatively, it is the parameterized version of the quantum complexity class Quantum Classical Merlin Arthur (\QCMA)~\cite{watrous2009quantum}.
\begin{definition}[\paraQCMA]
    \label{definition:paraQCMA}
    A parameterized problem $(L,\kappa)$ over the alphabet $\Sigma$ is in \paraQCMA{$(c,s)$} if there is a quantum verification procedure $\{\mathcal{V}_{n,k}\}_{n,k\in\mathbb{Z}^+}$ such that the following conditions are satisfied.
    \begin{enumerate}
        \item There is a computable function $f:\mathbb{Z}^+\to\mathbb{Z}^+$ and a polynomial $p\in\mathbb{N}[X]$, such that, for every $x\in\Sigma^*$, the size of an FPT-uniform quantum circuit that computes $\mathcal{V}_{\abs{x},\kappa(x)}$ on input $x$ is at most $f(\kappa(x)) \cdot p(\abs{x})$.
        \item For every $x\in\Sigma^*$,
        \begin{itemize}
            \item If $x \in L$, then there exists a bit string $y$, such that $\textbf{Pr}[\text{$\mathcal{V}_{\abs{x},\kappa(x)}(x,y)$ accepts}] \geq c$.
            \item If $x \notin L$, then for every bit string $y$, $\textbf{Pr}[\text{$\mathcal{V}_{\abs{x},\kappa(x)}(x,y)$ accepts}] \leq s$.
        \end{itemize}
    \end{enumerate}
    The class \paraQCMA is defined to be \paraQCMA{$(\frac{2}{3},\frac{1}{3})$}.
\end{definition}

We shall now present some alternative characterisations of \paraQCMA.
\begin{proposition}
    \label{proposition:ParaQCMAAlternativeCharacterisations}
    Let $(L,\kappa)$ be a parameterized problem over the alphabet $\Sigma$. Then the following statements are equivalent.
    \begin{enumerate}
        \item $(L,\kappa)$ is in \paraQCMA.
        \item $(L,\kappa)$ is in \QCMA after a precomputation on the parameter. That is, there exists an alphabet $\Pi$, a computable function $\pi:\mathbb{Z}^+\to\Pi^*$, and a problem $X\subseteq\Sigma^*\times\Pi^*$ such that $X$ is in \QCMA and, for all instances $x$ of $L$, we have $x \in L$ if and only if $(x,\pi(\kappa(x))) \in X$.
        \item $L$ is decidable and $(L,\kappa)$ is eventually in \QCMA. That is, there exists a computable function $\rho:\mathbb{Z}^+\to\mathbb{Z}^+$ and a QCMA algorithm that on input $x\in\Sigma^*$ with $\abs{x}\geq\rho(\kappa(x))$, decides if $x$ is a member of $L$ with error probability at most $1/3$.
    \end{enumerate}
\end{proposition}
\begin{proof}
    The proof follows similarly to that of Theorem~\ref{theorem:FPQTAlternativeCharacterisations}.
\end{proof}

The following containments are straightforward.
\begin{proposition}
    $\paraNP\subseteq\paraQCMA$, $\FPQT\subseteq\QCMA$, and $\paraQCMA\subseteq\paraQMA$.
\end{proposition}
We conjecture that these containments are strict, i.e., $\paraNP\neq\paraQCMA$, $\FPQT\neq\paraQCMA$, and $\paraQCMA\neq\paraQMA$. However, proving a separation between \paraNP and \paraQCMA is as difficult as proving a separation between \NP and \QCMA, proving a separation between \FPQT and \paraQCMA is as difficult as proving a separation between \BQP and \QCMA, and proving a separation between \paraQCMA and \paraQMA is as difficult as proving a separation between \QCMA and \QMA. Therefore, resolving these conjectures is a hard open problem.
\begin{proposition}
    $\paraNP=\paraQCMA$ if and only if $\NP=\QCMA$.
\end{proposition}
\begin{proof}
    The proof follows similarly to that of Proposition~\ref{proposition:ParaNPvsParaQMA}.
\end{proof}

\begin{proposition}
    \label{proposition:FPQTvsParaQCMA}
    $\FPQT=\paraQCMA$ if and only if $\BQP=\QCMA$.
\end{proposition}
\begin{proof}
    The proof follows similarly to that of Proposition~\ref{proposition:FPQTvsParaQMA}.
\end{proof}

\begin{proposition}
    $\paraQCMA=\paraQMA$ if and only if $\QCMA=\QMA$.
\end{proposition}
\begin{proof}
    Suppose that $\paraQCMA=\paraQMA$. For every problem $L\subseteq\Sigma^*$ in \QMA, we have that $(L,\kappa_\text{trivial})$ is in \paraQMA. It then follows that $(L,\kappa_\text{trivial})$ is in \paraQCMA and therefore $L$ is in \QCMA. Hence $\QCMA=\QMA$.

    Now suppose instead that $\QCMA=\QMA$, then the characterisation of \paraQCMA given by Proposition~\ref{proposition:ParaQCMAAlternativeCharacterisations} (2.) and the characterisation of \paraQMA given by Proposition~\ref{proposition:ParaQMAAlternativeCharacterisations} (2.) are equivalent. Hence $\paraQCMA=\paraQMA$.
\end{proof}

We shall now develop the theory of \paraQCMA-completeness. Firstly, we show that \paraQCMA is closed under FPQT reductions.
\begin{proposition}
    \paraQCMA is closed under FPQT reductions.
\end{proposition}
\begin{proof}
    The proof follows similarly to that of Proposition~\ref{proposition:FPQTClosedUnderFPQTReductions}.
\end{proof}

The following theorem allows us to establish the \paraQCMA-completeness of a wide range of problems.
\begin{theorem}
    \label{theorem:ParaQCMACompleteness}
    Let $(L,\kappa)$ be a non-trivial parameterized problem in \paraQCMA. Then the following statements are equivalent.
    \begin{enumerate}
        \item $(L,\kappa)$ is \textnormal{\paraQCMA-complete} under FPQT reductions.
        \item The union of a finite number of slices of $(L,\kappa)$ is \textnormal{\QCMA-complete}.
    \end{enumerate}
\end{theorem}
\begin{proof}
    The proof follows similarly to that of Theorem~\ref{theorem:ParaQMACompleteness}.
\end{proof}

We now proceed to show that the \textsc{Ground State $k$-Connectivity} problem is \paraQCMA-complete.
\begin{center}
\begin{tabularx}{\linewidth}{l X c}
    \multicolumn{2}{l}{\textsc{Ground State $k$-Connectivity}:} \\
    \textit{Instance:} & A local Hamiltonian $H\coloneqq\sum_iH_i$ on $n$ qubits, where each term $H_i$ has infinity norm $\norm{H_i}_\infty\leq1$. A polynomial $p\in\mathbb{N}[X]$. Two positive numbers $a,b\in(0,1)$, such that $b-a>\frac{1}{\text{poly}(n)}$. Two polynomial-size quantum circuits $U_\psi$ and $U_\phi$ generating states $\ket{\psi}=U_\psi\ket{0^n}$ and $\ket{\phi}=U_\phi\ket{0^n}$ such that $\expval{H}{\psi} \leq a$ and $\expval{H}{\phi} \leq a$. \\
    \textit{Parameter:} &  A natural number $k$. \\
    \textit{Problem:} & Decide whether there exists a sequence of $k$-local unitary matrices $(U_i)_{i=1}^{p(n)}$, such that:
    \begin{enumerate}
        \item For all $m\in[p(n)]$, the intermediate states $\ket{\psi_m}\coloneqq\prod_{i=1}^mU_i\ket{\psi}$ satisfy $\expval{H}{\psi_i} \leq a$.
        \item The final state $\ket{\psi_{p(n)}}\coloneqq\prod_{i=1}^{p(n)}U_i\ket{\phi}$ satisfies $\norm{\ket{\psi_{p(n)}}-\ket{\phi_{\vphantom{p(n)}}}}_2 \leq a$.
    \end{enumerate}
    Otherwise, if for all sequences of $k$-local unitary matrices $(U_i)_{i=1}^{p(n)}$, either:
    \begin{enumerate}
        \item There exists an $m\in[p(n)]$ and an intermediate state $\ket{\psi_m}\coloneqq\prod_{i=1}^mU_i\ket{\psi}$ such that $\expval{H}{\psi_i} \geq b$.
        \item The final state $\ket{\psi_{p(n)}}\coloneqq\prod_{i=1}^{p(n)}U_i\ket{\phi}$ satisfies $\norm{\ket{\psi_{p(n)}}-\ket{\phi_{\vphantom{p(n)}}}}_2 \geq b$.
    \end{enumerate}
    Given the promise that one of these is the case.
\end{tabularx}
\end{center}

\begin{corollary}
    \label{corollary:GroundStateKConnectivityParaQCMAComplete}
    \textsc{Ground State $k$-Connectivity} is \textnormal{\paraQCMA-complete}.
\end{corollary}
\begin{proof}
    The proof follows from Theorem~\ref{theorem:ParaQCMACompleteness} and the fact that \textsc{Ground State $k$-Connectivity} is \QCMA-complete for constant $k$~\cite{gharibian2018ground}.
\end{proof}

\subsection{Slice-Wise Quantum Polynomial Time}
\label{section:SliceWiseQuantumPolynomialTime}

The class \XP is often used in a similar way to \EXP (exponential time) as there is both a strict separation from \FPT under the time hierarchy, and the parameterizations of several \EXP-complete problems are \XP-complete. This includes the \textsc{Peg Game} when parameterized by the number of rings~\cite{downey1999parameterized} and the \textsc{Pebble Game} when parameterized by size of the start set~\cite{downey1999parameterized}. The \textsc{$n^k$-Step Halting Problem} for deterministic Turing machines is also \XP-complete when parameterized by $k$~\cite{flum2006parameterized}.

\begin{definition}[\XP]
    The class \XP consists of all parameterized problems $(L,\kappa)$ whose slices $(L,\kappa)_l$ for $l\geq1$ are all in \PT.
\end{definition}
  
This class provides a direct means for establishing that a problem is unlikely to be \paraNP-complete, as $\paraNP\subset\XP$ implies $\PT=\NP$~\cite[Proposition 2.20]{flum2006parameterized}. Thus membership in \XP is a useful tool for demonstrating the possibility of tractability. A similar theorem holds in the quantum case, providing a similar tool.

We shall now briefly introduce the complexity class Slice-wise Quantum Polynomial time (\XQP), which is the quantum analogue of the classical complexity class Slice-wise Polynomial time (\XP)~\cite{downey1999parameterized}.
\begin{definition}[\XQP]
    \label{definition:XQP}
    The class \XQP consists of all parameterized problems $(L,\kappa)$ whose slices $(L,\kappa)_l$ for $l\geq1$ are all in \BQP.
\end{definition}

The following containments are straightforward.
\begin{proposition}
    $\XP\subseteq\XQP$ and $\FPQT\subseteq\XQP$.
\end{proposition}
We conjecture that these containments are strict, i.e., $\XP\neq\XQP$ and $\FPQT\neq\XQP$. We also prove the following.
\begin{proposition}
    \label{proposition:FPQTinXPvsPinBQP}
    If $\FPQT\subseteq\XP$ then $\PT=\BQP$.
\end{proposition}
\begin{proof}
    If $\FPQT\subseteq\XP$ then any \BQP-complete problem with trivial parameterization is contained in \XP. Implying that $\PT=\BQP$.
\end{proof}

We shall now study the relationship between \XQP, \paraQMA, and \paraQCMA. We have the following propositions.
\begin{proposition}
    If $\BQP\neq\QMA$ then $\paraQMA\not\subseteq\XQP$.
\end{proposition}
\begin{proof}
    If $\paraQMA\subseteq\XQP$ then \textsc{$\kappa$-Local Hamiltonian} is in \XQP. Hence \textsc{$2$-Local Hamiltonian} is in \BQP, which implies $\BQP=\QMA$. Here we use the fact that \textsc{$2$-Local Hamiltonian} is \QMA-complete~\cite{kempe2006complexity}.
\end{proof}
\begin{proposition}
    If $\BQP\neq\QCMA$ then $\paraQCMA\not\subseteq\XQP$.
\end{proposition}
\begin{proof}
    If $\paraQCMA\subseteq\XQP$ then \textsc{Ground State $\kappa$-Connectivity} is in \XQP. Hence \textsc{Ground State $2$-Connectivity} is in \BQP, which implies $\BQP=\QCMA$. Here we use the fact that \textsc{Ground State $2$-Connectivity} is \QCMA-complete~\cite{gharibian2018ground}.
\end{proof}

The class \XQP is a non-uniform class and, in fact, contains problems that are undecidable. It is easy to see this because the class \XP contains problems that are undecidable~\cite{flum2006parameterized}. We define the following uniform version of \XQP.
\begin{definition}[$\XQP_\textsc{uniform}$]
    The class $\XQP_\textsc{uniform}$ consists of all parameterized problems $(L,\kappa)$ over the alphabet $\Sigma$ for which there is a computable function $f:\mathbb{Z}^+\to\mathbb{Z}^+$ and a quantum algorithm that, given $x\in\Sigma^*$, decides if $x$ is a member of $L$ with error probability at most $1/3$ and runs in time less than $\abs{x}^{f(\kappa(x))}+f(\kappa(x))$.
\end{definition}

\subsection{The Quantum Weft Hierarchy}
\label{section:QuantumWeftHierarchy}

The Weft hierarchy (\W) and its related class \WP are the central tools for demonstrating intractability in the classical parameterized setting. The \W hierarchy consists of an infinite hierarchy of classes $\W[t]$ for $t\in\mathbb{N}$ and is contained in the class \WP. The class \WP may be intuitively thought of as the subclass of \paraNP with sufficiently limited non-determinism to also be a subclass of \XP. The problem \textsc{Weight-$k$ Circuit Satisfiability} is complete for \WP, \textsc{$k$-Independent Set} is complete for $\W[1]$, and \textsc{$k$-Dominating Set} is complete for $\W[2]$~\cite{downey2013fundamentals}. Before introducing the quantum analogues of these classes, we review their classical definitions.

\begin{definition}[\WP]
    A parameterized problem $(L,\kappa)$ over the alphabet $\Sigma$ is in \WP if there is a verification procedure $\{\mathcal{V}_{n,k}\}_{n,k\in\mathbb{Z}^+}$ such that the following conditions are satisfied.
    \begin{enumerate}
        \item There is a computable function $f:\mathbb{Z}^+\to\mathbb{Z}^+$ and a polynomial $p\in\mathbb{N}[X]$, such that, for every $x\in\Sigma^*$, $\mathcal{V}_{\abs{x},\kappa(x)}$ on input $x$ runs in time at most $f(\kappa(x)) \cdot p(\abs{x})$ on a deterministic Turing machine.
        \item For every $x\in\Sigma^*$,
        \begin{itemize}
            \item If $x \in L$, then there exists a bit string $y$ comprising at most $f(\kappa(x))\cdot\log\abs{x}$ bits, such that $\mathcal{V}_{\abs{x},\kappa(x)}(x,y)$ accepts.
            \item If $x \notin L$, then for every bit string $y$ comprising at most $f(\kappa(x))\cdot\log\abs{x}$ bits, $\mathcal{V}_{\abs{x},\kappa(x)}(x,y)$ rejects.
        \end{itemize}
    \end{enumerate}
\end{definition}

To define the complexity class $\W[t]$, we require the notion of \emph{circuit weft}.
\begin{definition}[Circuit weft]
    Given a Boolean circuit $\mathcal{C}$ comprising generalised Toffoli gates and one and two bit fan-in gates. The \emph{weft} of $\mathcal{C}$ is the maximum number of Toffoli gates that act on any path from input bit to output bit.
\end{definition}

\begin{center}
\begin{tabularx}{\linewidth}{l X c}
    \multicolumn{2}{l}{\textsc{Weight-$k$ Weft-$t$ Depth-$d$ Circuit Satisfiability}:} \\
    \textit{Instance:} & A weft-$t$ depth-$d$ Boolean circuit $\mathcal{C}$ on $n$ input bits. \\
    \textit{Parameter:} &  A natural number $k$. \\
    \textit{Problem:} & Decide whether there exists an $n$-bit Hamming weight-$k$ string $y$, such that $\mathcal{C}(y)$ accepts.
\end{tabularx}
\end{center}

\begin{definition}[{$\W[t]$}]
    For $t\in\mathbb{N}$, the class $\W[t]$ consists of all parameterized problems that are FPT reducible to \textsc{Weight-$k$ Weft-$t$ Depth-$d$ Circuit Satisfiability} for some $d \geq t$.
\end{definition}

We shall now introduce the Quantum Weft hierarchy (\QW), which is the quantum version of the Weft hierarchy (\W). We begin by defining the complexity class \QWP --- the quantum version of the complexity class \WP.

\begin{definition}[\QWP]
    \label{definition:QWP}
    A parameterized problem $(L,\kappa)$ over the alphabet $\Sigma$ is in \QWP{$(c,s)$} if there is a quantum verification procedure $\{\mathcal{V}_{n,k}\}_{n,k\in\mathbb{Z}^+}$ such that the following conditions are satisfied.
    \begin{enumerate}
        \item There is a computable function $f:\mathbb{Z}^+\to\mathbb{Z}^+$ and a polynomial $p\in\mathbb{N}[X]$, such that, for every $x\in\Sigma^*$, the size of an FPT-uniform quantum circuit that computes $\mathcal{V}_{\abs{x},\kappa(x)}$ on input $x$ is at most $f(\kappa(x)) \cdot p(\abs{x})$.
        \item For every $x\in\Sigma^*$,
        \begin{itemize}
            \item If $x \in L$, then there exists a quantum state $\ket{\psi}$ comprising at most $f(\kappa(x))\cdot\log\abs{x}$ qubits, such that $\textbf{Pr}[\text{$\mathcal{V}_{\abs{x},\kappa(x)}(x,\ket{\psi})$ accepts}] \geq c$.
            \item If $x \notin L$, then for every quantum state $\ket{\psi}$ comprising at most $f(\kappa(x))\cdot\log\abs{x}$ qubits, $\textbf{Pr}[\text{$\mathcal{V}_{\abs{x},\kappa(x)}(x,\ket{\psi})$ accepts}] \leq s$.
        \end{itemize}
    \end{enumerate}
    The class \QWP is defined to be \QWP{$(\frac{2}{3},\frac{1}{3})$}.
\end{definition}

The following containments are straightforward.
\begin{proposition}
    $\WP\subseteq\QWP$ and $\FPQT\subseteq\QWP$.
\end{proposition}
We conjecture that these containments are strict, i.e., $\WP\neq\QWP$ and $\FPQT\neq\QWP$. However, proving a separation between \FPQT and \QWP is as difficult as proving a separation between \BQP and \QMA.
\begin{proposition}
    \label{proposition:FPQTvsQWP}
    If $\FPQT\neq\QWP$ then $\BQP\neq\QMA$.
\end{proposition}
\begin{proof}
    Since $\FPQT\subseteq\QWP\subseteq\paraQMA$, we have that $\FPQT\neq\QWP$ implies $\FPQT\neq\paraQMA$. It then follows from Proposition~\ref{proposition:FPQTvsParaQMA}, that $\BQP\neq\QMA$.
\end{proof}

We also have the following containment.
\begin{proposition}
    \label{proposition:QWPinXQPcapParaQMA}
    $\QWP\subseteq\XQP\cap\paraQMA$.
\end{proposition}
\begin{proof}
    It is straightforward to prove that $\QWP\subseteq\paraQMA$. To prove that $\QWP\subseteq\XQP$, observe that any problem in $\QWP$ can be solved in quantum time $O\left(2^{f(\kappa(x))\cdot\log\abs{x}}\right)=\abs{x}^{O(f(\kappa(x)))}$, see, for example, Section~\ref{section:QuantumMerlinArthurProofs}.
\end{proof}

We shall now develop the theory of \QWP-completeness. Firstly, we show that \QWP is closed under FPQT reductions.
\begin{proposition}
    \QWP is closed under FPQT reductions.
\end{proposition}
\begin{proof}
    The proof follows similarly to that of Proposition~\ref{proposition:FPQTClosedUnderFPQTReductions}.
\end{proof}

We now introduce the \textsc{Weight-$k$ Quantum Circuit Satisfiability} problem and show that this problem is \QWP-complete. This requires the notion of the \emph{weight of a quantum state}.
\begin{definition}[Weight of a quantum state]
    A quantum state $\ket{\psi}=\sum_{x\in\{0,1\}^n}\alpha_x\ket{x}$ on $n$ qubits is said to have \emph{weight} $k$ if $\alpha_x=0$ for all $x$ not of Hamming weight $k$.
\end{definition}

\begin{center}
\begin{tabularx}{\linewidth}{l X c}
    \multicolumn{2}{l}{\textsc{Weight-$k$ Quantum Circuit Satisfiability}:} \\
    \textit{Instance:} & A quantum circuit $\mathcal{C}$ on $n$ witness qubits and $\text{poly}(n)$ ancilla qubits. Two positive numbers $a,b\in(0,1)$, such that $b-a>\frac{1}{\text{poly}(n)}$. \\
    \textit{Parameter:} &  A natural number $k$. \\
    \textit{Problem:} & Decide whether there exists an $n$-qubit weight-$k$ quantum state $\ket{\psi}$, such that $\textbf{Pr}[\text{$\mathcal{C}(\ket{\psi})$ accepts}] \geq b$. Otherwise, if for every $n$-qubit weight-$k$ quantum state $\ket{\psi}$, $\textbf{Pr}[\text{$\mathcal{C}(\ket{\psi})$ accepts}] \leq a$. Given the promise that one of these is the case.
\end{tabularx}
\end{center}

We now establish our completeness result.
\begin{proposition}
    \label{proposition:WeightKQuantumCircuitSatisfiabilityQWPComplete}
    \textsc{Weight-$k$ Quantum Circuit Satisfiability} is \textnormal{\QWP-complete} under FPQT reductions.
\end{proposition}
\begin{proof}
    Firstly, we show that \textsc{Weight-$k$ Quantum Circuit Satisfiability} is in \QWP. Let $\mathcal{C}$ be a quantum circuit on $n$ qubits, $k$ a natural number, and $f:\mathbb{Z}^+\to\mathbb{Z}^+$ a computable function. Further let $S_{n,k}$ denote the set of all $n$-bit strings with Hamming weight $k$ and let $\varepsilon$ be a binary enumeration of the elements of $S_{n,k}$. An $n$-qubit weight-$k$ quantum state $\ket{\psi}=\sum_{x \in S_{n,k}}\alpha_x\ket{x}$ can be described using $f(k)\cdot\log(n)$ qubits by the quantum state $\ket{\psi_\varepsilon}=\sum_{x \in S_{n,k}}\alpha_x\ket{\varepsilon(x)}$. Let $\mathcal{M}_{n,k}$ be a verification procedure for deciding whether the a weight of an $n$-qubit quantum state is $k$. The verification procedure $\mathcal{V}_{n,k}$ constructs the state $\ket{\psi}$ from $\ket{\psi_\varepsilon}$ and accepts if and only if $\mathcal{C}(\mathcal{M}_{n,k}\ket{\psi})$ accepts. Applying the gap amplification scheme of Marriott and Watrous~\cite{marriott2005quantum} to this procedure completes the claim.
    
    We now prove that \textsc{Weight-$k$ Quantum Circuit Satisfiability} is \QWP-hard. Let $(L,\kappa)$ be a problem in \QWP with verification procedure $\{\mathcal{V}_{n,k}\}_{n,k\in\mathbb{Z}^+}$. Further let $f:\mathbb{Z}^+\to\mathbb{Z}^+$ be a computable function and define $k_x \coloneqq \kappa(x)$. For input $x\in\Sigma^*$, we shall construct a quantum circuit $\mathcal{C}_x$ that is satisfiable by a weight-$k_x$ quantum state if and only if $\mathcal{V}_{\abs{x},\kappa(x)}(x)$ is satisfiable. The circuit $\mathcal{C}_x$ takes as input $n$ qubits and firstly decides whether the input state has weight $k_x$ using the verification procedure $\mathcal{M}_{n,k_x}$. Finally, the circuit inputs the quantum state into the verifier $\mathcal{V}_{\abs{x},\kappa(x)}$. Therefore, $\mathcal{C}_x$ is satisfiable by a weight-$k_x$ quantum state if and only if $\mathcal{V}_{\abs{x},\kappa(x)}(x)$ is satisfiable. This completes the proof.
\end{proof}

It is natural to ask whether the \textsc{$l$-Local Hamiltonian} variant of this problem is \QWP-complete. However, as we shall see this problem is in \XP. Note that the slices of the \textsc{Weight-$k$ Quantum Circuit Satisfiability} problem are \BQP-complete and so it cannot be in \XP unless $\PT=\BQP$.

\begin{center}
\begin{tabularx}{\linewidth}{l X c}
    \multicolumn{2}{l}{\textsc{Weight-$k$ $l$-Local Hamiltonian}:} \\
    \textit{Instance:} & An $l$-local Hamiltonian $H\coloneqq\sum_iH_i$ on $n$ qubits that comprises at most a polynomial in $n$ many terms $\{H_i\}$, which each act non-trivially on at most $l$ qubits and have operator norm $\norm{H_i}$ bounded from above by a polynomial in $n$. Two positive numbers $a,b\in(0,1)$, such that $b-a>\frac{1}{\text{poly}(n)}$. \\
    \textit{Parameter:} &  A natural number $k$. \\
    \textit{Problem:} & Decide whether there exists an $n$-qubit weight-$k$ quantum state $\ket{\psi}$, such that $\matrixel{\psi}{H}{\psi} \leq a$. Otherwise, if for every $n$-qubit weight-$k$ quantum state $\ket{\psi}$, $\matrixel{\psi}{H}{\psi} \geq b$. Given the promise that one of these is the case.
\end{tabularx}
\end{center}

\begin{proposition}
    \label{proposition:WeightLocalHamiltonianinXP}
    \textsc{Weight-$k$ $l$-Local Hamiltonian} is in \XP.
\end{proposition}
\begin{proof}
    Let $S_{n,k}$ denote the set of all $n$-bit strings with Hamming weight $k$ and let $\varepsilon$ be an enumeration of the elements of $S_{n,k}$. We define the matrix $H_\varepsilon$ such that $\matrixel{\varepsilon(x)}{H_\varepsilon}{\varepsilon(y)}\coloneqq\matrixel{x}{H}{y}$ for all $x,y \in S_{n,k}$, and for an $n$-qubit weight-$k$ quantum state $\ket{\psi}\coloneqq\sum_{x \in S_{n,k}}\alpha_x\ket{x}$, we define the quantum state $\ket{\psi_\varepsilon}\coloneqq\sum_{x \in S_{n,k}}\alpha_x\ket{\varepsilon(x)}$. Then, for any $n$-qubit weight-$k$ quantum state $\ket{\psi}$, we have $\matrixel{\psi}{H}{\psi}=\matrixel{\psi_\varepsilon}{H_\varepsilon}{\psi_\varepsilon}$. Therefore, it is sufficient to compute the smallest eigenvalue $\lambda_{\text{min}}(H_\varepsilon)$ of $H_\varepsilon$. However, since the dimension of $H_\varepsilon$ is $n^{O(k)}$ and each of its entries can be computed in time $n^{O(1)}$, we can compute $\lambda_{\text{min}}(H_\varepsilon)$ in time $n^{O(k)}$. Hence, \textsc{Weight-$k$ $l$-Local Hamiltonian} is in \XP. This completes the proof.
\end{proof}

We prove the following.
\begin{proposition}
    \label{proposition:QWPinXPvsPequalsBQP}
    If $\QWP\subseteq\XP$ then $\PT=\BQP$.
\end{proposition}
\begin{proof}
    If $\QWP\subseteq\XP$ then $\FPQT\subseteq\XP$, and so $\PT=\BQP$ by Proposition~\ref{proposition:FPQTinXPvsPinBQP}.
\end{proof}

We shall now define the complexity class $\QW[t]$ --- the quantum version of the complexity class $\W[t]$. This requires the notion of \emph{quantum circuit weft}.
\begin{definition}[Quantum circuit weft]
    Given a quantum circuit $\mathcal{C}$ comprising generalised Toffoli gates, one and two-qubit gates, and unbounded classical fan-out. The \emph{weft} of $\mathcal{C}$ is the maximum number of Toffoli gates that act on any path from input qubit to output qubit.
\end{definition}

\begin{center}
\begin{tabularx}{\linewidth}{l X c}
    \multicolumn{2}{l}{\textsc{Weight-$k$ Weft-$t$ Depth-$d$ Quantum Circuit Satisfiability}:} \\
    \textit{Instance:} & A weft-$t$ depth-$d$ quantum circuit $\mathcal{C}$ on $n$ witness qubits and $\text{poly}(n)$ ancilla qubits. Three positive numbers $a,b,c\in(0,1)$, such that $b-a>c$. \\
    \textit{Parameter:} &  A natural number $k$. \\
    \textit{Problem:} & Decide whether there exists an $n$-qubit weight-$k$ quantum state $\ket{\psi}$, such that $\textbf{Pr}[\text{$\mathcal{C}(\ket{\psi})$ accepts}] \geq b$. Otherwise, if for every $n$-qubit weight-$k$ quantum state $\ket{\psi}$, $\textbf{Pr}[\text{$\mathcal{C}(\ket{\psi})$ accepts}] \leq a$. Given the promise that one of these is the case.
\end{tabularx}
\end{center}

\begin{definition}[{$\QW[t]$}]
    \label{definition:QWT}
    For $t\in\mathbb{N}$, the class $\QW[t]$ consists of all parameterized problems that are FPQT reducible to \textsc{Weight-$k$ Weft-$t$ Depth-$d$ Quantum Circuit Satisfiability} for some $d \geq t$.
\end{definition}

The following containments are straightforward.
\begin{proposition}
    For any $t\in\mathbb{N}$, $\W[t]\subseteq\QW[t]$, $\QW[t]\subseteq\QW[t+1]$, and $\QW[t]\subseteq\QWP$.
\end{proposition}
The complexity classes $\QW[t]$, for $t\geq1$, define the \QW hierarchy, while note that $\QW[0]=\FPQT$. We prove the following.
\begin{proposition}
    For any $t\in\mathbb{N}$, if $\W[t]=\QW[t]$ then $\PT=\BQP$.
\end{proposition}
\begin{proof}
    If $\W[t]=\QW[t]$ then $\FPQT\subseteq\XP$, and so $\PT=\BQP$ by Proposition~\ref{proposition:FPQTinXPvsPinBQP}.
\end{proof}
\begin{proposition}
    For any $t\in\mathbb{N}$, if $\QW[t]\subseteq\XP$ then $\PT=\BQP$.
\end{proposition}
\begin{proof}
    If $\QW[t]\subseteq\XP$ then $\FPQT\subseteq\XP$, and so $\PT=\BQP$ by Proposition~\ref{proposition:FPQTinXPvsPinBQP}.
\end{proof}

\subsection{The Quantum Classical Weft Hierarchy}
\label{section:QuantumClassicalWeftHierarchy}

We shall now introduce the Quantum Classical Weft hierarchy (\QCW). We begin by define the complexity class \QCWP, which is the subclass of \QWP restricted to classical proofs.
\begin{definition}[\QCWP]
    \label{definition:QCWP}
    A parameterized problem $(L,\kappa)$ over the alphabet $\Sigma$ is in \QCWP{$(c,s)$} if there is a quantum verification procedure $\{\mathcal{V}_{n,k}\}_{n,k\in\mathbb{Z}^+}$ such that the following conditions are satisfied.
    \begin{enumerate}
        \item There is a computable function $f:\mathbb{Z}^+\to\mathbb{Z}^+$ and a polynomial $p\in\mathbb{N}[X]$, such that, for every $x\in\Sigma^*$, the size of an FPT-uniform quantum circuit that computes $\mathcal{V}_{\abs{x},\kappa(x)}$ on input $x$ is at most $f(\kappa(x)) \cdot p(\abs{x})$.
        \item For every $x\in\Sigma^*$,
        \begin{itemize}
            \item If $x \in L$, then there exists a bit string $y$ comprising at most $f(\kappa(x))\cdot\log\abs{x}$ bits, such that $\textbf{Pr}[\text{$\mathcal{V}_{\abs{x},\kappa(x)}(x,y)$ accepts}] \geq c$.
            \item If $x \notin L$, then for every bit string $y$ comprising at most $f(\kappa(x))\cdot\log\abs{x}$ bits, $\textbf{Pr}[\text{$\mathcal{V}_{\abs{x},\kappa(x)}(x,y))$ accepts}] \leq s$.
        \end{itemize}
    \end{enumerate}
    The class \QCWP is defined to be \QCWP{$(\frac{2}{3},\frac{1}{3})$}.
\end{definition}

The following containments are straightforward.
\begin{proposition}
    $\WP\subseteq\QCWP$, $\FPQT\subseteq\QCWP$, and $\QCWP\subseteq\QWP$.
\end{proposition}
We conjecture that these containments are strict, i.e., $\WP\neq\QCWP$, $\FPQT\neq\QCWP$, and $\QCWP\neq\QWP$. Similar to Proposition~\ref{proposition:FPQTvsQWP}, proving a separation between \FPQT and \QCWP is as difficult as proving a separation between \BQP and \QCMA.
\begin{proposition}
    If $\FPQT\neq\QCWP$ then $\BQP\neq\QCMA$.
\end{proposition}
\begin{proof}
    Since $\FPQT\subseteq\QCWP\subseteq\paraQCMA$, we have that $\FPQT\neq\QCWP$ implies $\FPQT\neq\paraQCMA$. It then follows from Proposition~\ref{proposition:FPQTvsParaQCMA}, that $\BQP\neq\QCMA$.
\end{proof}
We also have the following containment.
\begin{proposition}
    $\QCWP\subseteq\XQP\cap\paraQCMA$.
\end{proposition}
\begin{proof}
    The proof follows similarly to that of Proposition~\ref{proposition:QWPinXQPcapParaQMA}.
\end{proof}

We shall now develop the theory of \QCWP-completeness. Firstly, we show that \QCWP is closed under FPQT reductions.
\begin{proposition}
    \QCWP is closed under FPQT reductions.
\end{proposition}
\begin{proof}
    The proof follows similarly to that of Proposition~\ref{proposition:FPQTClosedUnderFPQTReductions}.
\end{proof}

We now introduce the \textsc{Weight-$k$ Quantum Circuit Satisfiability} problem and show that this problem is \QWP-complete. This requires the notion of the \emph{weight of a quantum state}.
\begin{definition}[Classical weight of a quantum state]
    A quantum state $\ket{\psi}$ on $n$ qubits is said to have \emph{classical weight} $k$ if there exists a quantum state $\ket{\phi}$ on $k$ qubits and an $n$-qubit swap network $\mathcal{S}_n$, such that $\ket{\psi}=\mathcal{S}_n(\ket{\phi}\ket{0^{n-k}})$.
\end{definition}

\begin{center}
\begin{tabularx}{\linewidth}{l X c}
    \multicolumn{2}{l}{\textsc{Classical Weight-$k$ Quantum Circuit Satisfiability}:} \\
    \textit{Instance:} & A quantum circuit $\mathcal{C}$ on $n$ witness qubits and $\text{poly}(n)$ ancilla qubits. Two positive numbers $a,b\in(0,1)$, such that $b-a>\frac{1}{\text{poly}(n)}$. \\
    \textit{Parameter:} & A natural number $k$. \\
    \textit{Problem:} & Decide whether there exists an $n$-qubit classical weight-$k$ quantum state $\ket{\psi}$, such that $\textbf{Pr}[\text{$\mathcal{C}(\ket{\psi})$ accepts}] \geq b$. Otherwise, if for every $n$-qubit classical weight-$k$ quantum state $\ket{\psi}$, $\textbf{Pr}[\text{$\mathcal{C}(\ket{\psi})$ accepts}] \leq a$. Given the promise that one of these is the case.
\end{tabularx}
\end{center}

\begin{proposition}
    \label{proposition:ClassicalWeightKQuantumCircuitSatisfiabilityQCWPComplete}
    \textsc{Classical Weight-$k$ Quantum Circuit Satisfiability} is \textnormal{\QCWP-complete}.
\end{proposition}
\begin{proof}
    Firstly, we show that \textsc{Classical Weight-$k$ Quantum Circuit Satisfiability} is in \QCWP. Let $\mathcal{C}$ be a quantum circuit on $n$ qubits, $k$ a natural number, and $f:\mathbb{Z}^+\to\mathbb{Z}^+$ a computable function. An $n$-qubit classical weight-$k$ quantum state $\ket{\psi}$ can be described to an arbitrary constant precision $\epsilon>0$ using $f(k)\cdot\log(n)$ classical bits by specifying a quantum circuit on $k$ qubits $\mathcal{D}_k$ and an $n$-qubit swap network $\mathcal{S}_n$, such that $\norm{\ket{\psi}-\mathcal{S}_n\mathcal{D}_k\ket{0^n}}_2\leq\epsilon$. The verification procedure $\mathcal{V}_{n,k}$ constructs the state $\mathcal{S}_n\mathcal{D}_k\ket{0^n}$ and accepts if and only if $\mathcal{C}(\mathcal{S}_n\mathcal{D}_k\ket{0^n})$ accepts. Applying the gap amplification scheme of Marriott and Watrous~\cite{marriott2005quantum} to this procedure completes the claim.

    We now prove that \textsc{Classical Weight-$k$ Quantum Circuit Satisfiability} is \QCWP-hard. Let $(L,\kappa)$ be a problem in \QCWP with verification procedure $\{\mathcal{V}_{n,k}\}_{n,k\in\mathbb{Z}^+}$. Further let $f:\mathbb{Z}^+\to\mathbb{Z}^+$ be a computable function and define $k_x \coloneqq f(\kappa(x))+\kappa(x)$. For input $x\in\Sigma^*$, we shall construct a quantum circuit $\mathcal{C}_x$ that is satisfiable by a classical weight-$k_x$ quantum state if and only if $\mathcal{V}_{\abs{x},\kappa(x)}(x)$ is satisfiable. The circuit $\mathcal{C}_x$ takes as input $f(\kappa(x))+\kappa(x)\cdot\abs{x}$ bits arranged in one block of $f(\kappa(x))$ bits and $\kappa(x)$ blocks of $\abs{x}$ bits. Note that a classical input can be ensured by computational basis measurements at the beginning of the circuit. The circuit $\mathcal{C}_x$ now decides whether each of the $\kappa(x)$ blocks of $\abs{x}$ bits has Hamming weight exactly one. If so, then the input has classical weight $k_x$. Otherwise, we reject. Each of these blocks is then mapped to a block of $\log\abs{x}$ bits that specifies the location of the one in the block. Finally, the circuit inputs the block of $f(\kappa(x))$ bits and the $\kappa(x)$ blocks of $\log\abs{x}$ bits into the verifier $\mathcal{V}_{\abs{x},\kappa(x)}$. Therefore, $\mathcal{C}_x$ is satisfiable by a classical weight-$k_x$ quantum state if and only if $\mathcal{V}_{\abs{x},\kappa(x)}(x)$ is satisfiable. This completes the proof.
\end{proof}

We shall now define the complexity class $\QCW[t]$, which is the subclass of $\QW[t]$ restricted to classical weights.
\begin{center}
\begin{tabularx}{\linewidth}{l X c}
    \multicolumn{2}{l}{\textsc{Classical Weight-$k$ Weft-$t$ Depth-$d$ Quantum Circuit Satisfiability}:} \\
    \textit{Instance:} & A weft-$t$ depth-$d$ quantum circuit $\mathcal{C}$ on $n$ witness qubits and $\text{poly}(n)$ ancilla qubits. Three positive numbers $a,b,c\in(0,1)$, such that $b-a>c$. \\
    \textit{Parameter:} &  A natural number $k$. \\
    \textit{Problem:} & Decide whether there exists an $n$-qubit classical weight-$k$ quantum state $\ket{\psi}$, such that $\textbf{Pr}[\text{$\mathcal{C}(\ket{\psi})$ accepts}] \geq b$. Otherwise, if for every $n$-qubit classical weight-$k$ quantum state $\ket{\psi}$, $\textbf{Pr}[\text{$\mathcal{C}(\ket{\psi})$ accepts}] \leq a$. Given the promise that one of these is the case.
\end{tabularx}
\end{center}

\begin{definition}[{$\QCW[t]$}]
    \label{definition:QCWT}
    For $t\in\mathbb{N}$, the class $\QCW[t]$ consists of all parameterized problems that are FPQT reducible to \textsc{Classical Weight-$k$ Weft-$t$ Depth-$d$ Quantum Circuit Satisfiability} for some $d \geq t$.
\end{definition}

The following containments are straightforward.
\begin{proposition}
    For any $t\in\mathbb{N}$, $\W[t]\subseteq\QCW[t]$, $\QCW[t]\subseteq\QCW[t+1]$, $\QCW[t]\subseteq\QW[t]$, and $\QCW[t]\subseteq\QCWP$.
\end{proposition}
The complexity classes $\QCW[t]$, for $t\geq1$, define the \QCW hierarchy, while note that, for $t=0$, we have $\QCW[0]=\FPQT$. We now establish an alternate \QCWP-complete problem, which leads to an variation of the \QCW hierarchy.
\begin{center}
\begin{tabularx}{\linewidth}{l X c}
    \multicolumn{2}{l}{\textsc{Hamming Weight-$k$ Quantum Circuit Satisfiability}:} \\
    \textit{Instance:} & A quantum circuit $\mathcal{C}$ on $n$ witness qubits and $\text{poly}(n)$ ancilla qubits. Two positive numbers $a,b\in(0,1)$, such that $b-a>\frac{1}{\text{poly}(n)}$. \\
    \textit{Parameter:} &  A natural number $k$. \\
    \textit{Problem:} & Decide whether there exists an $n$-bit Hamming weight-$k$ string $y$, such that $\textbf{Pr}[\text{$\mathcal{C}(y)$ accepts}] \geq b$. Otherwise, if for every $n$-bit Hamming weight-$k$ string $y$, $\textbf{Pr}[\text{$\mathcal{C}(y)$ accepts}] \leq a$. Given the promise that one of these is the case.
\end{tabularx}
\end{center}

\begin{proposition}
    \textsc{Hamming Weight-$k$ Quantum Circuit Satisfiability} is \textnormal{\QCWP-complete}.
\end{proposition}
\begin{proof}
    The proof follows similarly to that of Proposition~\ref{proposition:ClassicalWeightKQuantumCircuitSatisfiabilityQCWPComplete}.
\end{proof}

\begin{center}
\begin{tabularx}{\linewidth}{l X c}
    \multicolumn{2}{l}{\textsc{Hamming Weight-$k$ Weft-$t$ Depth-$d$ Quantum Circuit Satisfiability}:} \\
    \textit{Instance:} & A weft-$t$ depth-$d$ quantum circuit $\mathcal{C}$ on $n$ witness qubits and $\text{poly}(n)$ ancilla qubits. Three positive numbers $a,b,c\in(0,1)$, such that $b-a>c$. \\
    \textit{Parameter:} &  A natural number $k$. \\
    \textit{Problem:} & Decide whether there exists an $n$-bit Hamming weight-$k$ string $y$, such that $\textbf{Pr}[\text{$\mathcal{C}(y)$ accepts}] \geq b$. Otherwise, if for every $n$-bit Hamming weight-$k$ string $y$, $\textbf{Pr}[\text{$\mathcal{C}(y)$ accepts}] \leq a$. Given the promise that one of these is the case.
\end{tabularx}
\end{center}

\begin{definition}[{$\QCWc[t]$}]
    \label{definition:QCWCT}
    For $t\in\mathbb{N}$, the class $\QCWc[t]$ consists of all parameterized problems that are FPQT reducible to \textsc{Hamming Weight-$k$ Weft-$t$ Depth-$d$ Quantum Circuit Satisfiability} for some $d \geq t$.
\end{definition}

The following containments are straightforward.
\begin{proposition}
    For any $t\in\mathbb{N}$, $\W[t]\subseteq\QCWc[t]$, $\QCWc[t]\subseteq\QCWc[t+1]$, $\QCWc[t]\subseteq\QCW[t]$, and $\QCWc[t]\subseteq\QCWP$.
\end{proposition}
The complexity classes $\QCWc[t]$, for $t\geq1$, define the \QCWc hierarchy, while note that, for $t=0$, we have $\QCWc[0]=\QCW[0]=\FPQT$.

\section{Applications}
\label{section:Applications}

In this section we shall explore some applications of fixed-parameter quantum tractability.

\subsection{Approximate Counting}
\label{section:ApproximateCounting}

In this section we shall apply some standard arguments in quantum computation and approximate counting to the parameterized setting. We proceed by establishing an FPQT algorithm for approximating quantum probability amplitudes.

\begin{theorem}
    \label{theorem:AdditiveErrorApproximationAlgorithmAmplitudes}
    Fix $\epsilon>0$. Let $\{\mathcal{C}_n\}_{n\in\mathbb{Z}^+}$ be a polynomial-time uniform family of quantum circuits each acting on $p(n)$ qubits and define $q(x)\coloneqq\matrixel{0^{p(\abs{x})}}{\mathcal{C}_{\abs{x}}(x)}{0^{p(\abs{x})}}$. Further let $\kappa$ be a parameterization and $f:\mathbb{Z}^+\to\mathbb{Z}^+$ a computable function. Then there is an FPQT algorithm that, for any input $x$, outputs an approximation to $q(x)$ to within an additive error of $\frac{\epsilon}{f(\kappa(x))\cdot\mathrm{poly}(\abs{x})}$.
\end{theorem}

\begin{proof}
    We apply the Hadamard test to $\mathcal{C}_{\abs{x}}(x)$ to sample from random variables with expectation values equal to $\Re(q(x))$ and $\Im(q(x))$. It follows from the Chernoff-Hoeffding bound that repeating this procedure $f(\kappa(x))^2\cdot\text{poly}(\abs{x})\cdot\epsilon^{-2}$ times allows us to obtain an approximation $\tilde{q}(x)$ to $q(x)$, such that
    \begin{equation}
        \textbf{Pr}\left[\abs{q(x)-\tilde{q}(x)} \geq \frac{\epsilon}{f(\kappa(x))\cdot\text{poly}(\abs{x})}\right] \leq e^{-\text{poly}(\abs{x})}. \notag
    \end{equation}
    Then, with high probability, we have
    \begin{equation}
        \abs{q(x)-\tilde{q}(x)} \leq \frac{\epsilon}{f(\kappa(x))\cdot\text{poly}(\abs{x})}. \notag
    \end{equation}
    This completes the proof.
\end{proof}

We obtain the immediate corollary.
\begin{corollary}
    \label{corollary:MultiplicativeErrorApproximationAlgorithmAmplitudes}
    In the notation of Theorem~\ref{theorem:AdditiveErrorApproximationAlgorithmAmplitudes}. Let $\kappa$ be a parameterization such that, for any input $x$,
    $\abs{q(x)}\geq\frac{1}{f(\kappa(x))\cdot\mathrm{poly}(\abs{x})}$. Then there is an FPQT algorithm that, for any input $x$, outputs a multiplicative $\epsilon$-approximation to $q(x)$.
\end{corollary}

It is well known that quantum probability amplitudes encode the evaluation of Jones polynomials at \emph{principal non-lattice roots of unity}, i.e., $t=\exp(2\pi i/k)$ for $k=5$ or $k\geq7$~\cite{freedman2002modular, aharonov2009polynomial}. This is the key observation used to establish the quantum algorithm of Aharonov, Jones, and Landau~\cite{aharonov2009polynomial} for approximating the evaluation of such Jones polynomials. We shall apply Theorem~\ref{theorem:AdditiveErrorApproximationAlgorithmAmplitudes} and Corollary~\ref{corollary:MultiplicativeErrorApproximationAlgorithmAmplitudes} to extend this algorithm to the parameterized setting. Recall that a \emph{braid} is a collection of strands that may cross over and under each other, and must always move from left to right. The \emph{plat closure} of a $2n$-strand braid $b$ is the link formed by connecting pairs of adjacent strands on the left and the right of the braid. We have the following corollary.

\begin{corollary}
    Fix $\epsilon>0$. Let $k=5$ or $k\geq7$ be an integer, and $t=\exp(2\pi i/k)$ its corresponding root of unity. Let $b$ be a braid on $2n$ strands with at most a polynomial in $n$ number of crossings, and let $b^{pl}$ denote its plat closure. Further let $\kappa$ be a parameterization and $f:\mathbb{Z}^+\to\mathbb{Z}^+$ a computable function. Then there is an FPQT algorithm that outputs an approximation to the evaluation of the Jones polynomial $\mathrm{V}_{b^{pl}}(t)$ to within an additive error of $\frac{\epsilon\cdot(2\cos(\pi/k))^n}{f(\kappa(b))\cdot\mathrm{poly}(\abs{x})}$. Furthermore, if $\kappa$ is a parameterization such that $\abs{\mathrm{V}_{b^{pl}}(t)}\geq\frac{(2\cos(\pi/k))^n}{f(\kappa(x))\cdot\mathrm{poly}(\abs{x})}$, then the FPQT algorithm outputs a multiplicative $\epsilon$-approximation to $\mathrm{V}_{b^{pl}}(t)$.
\end{corollary}

\begin{proof}
    By a result of Aharonov, Jones, and Landau~\cite{aharonov2009polynomial} there exists a polynomial-time uniform family of quantum circuits $\{\mathcal{C}_{n,k}\}_{n,k\in\mathbb{Z}^+}$ and a polynomial $p\in\mathbb{N}[X]$ such that, for any braid $b$ on $2n$ strands with at most a polynomial in $n$ number of crossings,
    \begin{equation}
        \matrixel{0^{p(\abs{x})}}{C_{n,k}(b)}{0^{p(\abs{x})}} = \frac{e^{\frac{3i\pi(k+1)}{2k}w(b^{pl})}}{(2\cos(\pi/k))^{n-1}}\mathrm{V}_{b^{pl}}(t), \notag
    \end{equation}
    where $w(b^{pl})$ is the \emph{writhe} of $b^{pl}$, which can be computed in polynomial time. By applying Theorem~\ref{theorem:AdditiveErrorApproximationAlgorithmAmplitudes}, we have an FPQT algorithm that outputs an approximation to $\mathrm{V}_{b^{pl}}(t)$ to within an additive error of $\frac{\epsilon\cdot(2\cos(\pi/k))^n}{f(\kappa(b))\cdot\text{poly}(\abs{x})}$. Then by Corollary~\ref{corollary:MultiplicativeErrorApproximationAlgorithmAmplitudes}, if $\kappa$ is a parameterization such that $\abs{\mathrm{V}_{b^{pl}}(t)}\geq\frac{(2\cos(\pi/k))^n}{f(\kappa(x))\cdot\mathrm{poly}(\abs{x})}$, then the FPQT algorithm outputs a multiplicative $\epsilon$-approximation to $\mathrm{V}_{b^{pl}}(t)$. This completes the proof.
\end{proof}

Fenner et al.~\cite{fenner1999determining} showed that the solution to arbitrary problems in \GapP can be efficiently encoded in quantum probability amplitudes. Recall that \GapP is the closure of \SharpP under subtraction. This allows us to establish an FPQT algorithm for approximating the solution to arbitrary problems in \GapP. However, as we shall see later, this can be achieved by a classical FPT algorithm.

\begin{corollary}
    Fix $\epsilon>0$. Let $g$ be a function in \GapP, and $p\in\mathbb{N}[X]$ a polynomial such that, for any input $x$, $g(x)$ takes values in the range $[-2^{p(\abs{x})},\;2^{p(\abs{x})}]$. Further let $\kappa$ be a parameterization and $f:\mathbb{Z}^+\to\mathbb{Z}^+$ a computable function. Then there is an FPQT algorithm that, for any input $x$, outputs an approximation to $g(x)$ to within an additive error of $\frac{\epsilon\cdot2^{p(\abs{x})}}{f(\kappa(x))\cdot\mathrm{poly}(\abs{x})}$. Furthermore, if $\kappa$ is a parameterization such that $\abs{g(x)}\geq\frac{2^{p(\abs{x})}}{f(\kappa(x))\cdot\mathrm{poly}(\abs{x})}$, then the FPQT algorithm outputs a multiplicative $\epsilon$-approximation to $g(x)$.
\end{corollary}

\begin{proof}
    By a result of Fenner et al.~\cite[Theorem 3.2]{fenner1999determining} there exists a polynomial-time uniform family of quantum circuits $\{\mathcal{C}_n\}_{n\in\mathbb{Z}^+}$ and a polynomial $q\in\mathbb{N}[X]$ such that, for all $x$ of length $n$,
    \begin{equation}
        \matrixel{0^{q(\abs{x})}}{\mathcal{C}_{\abs{x}}(x)}{0^{q(\abs{x})}} = \frac{g(x)}{2^{p(\abs{x})}}. \notag
    \end{equation}
    By applying Theorem~\ref{theorem:AdditiveErrorApproximationAlgorithmAmplitudes}, we have an FPQT algorithm that outputs an approximation to $g(x)$ to within an additive error of $\frac{\epsilon\cdot2^{p(\abs{x})}}{f(\kappa(b))\cdot\text{poly}(\abs{x})}$. Then by Corollary~\ref{corollary:MultiplicativeErrorApproximationAlgorithmAmplitudes}, if $\kappa$ is a parameterization such that $\abs{g(x)}\geq\frac{2^{p(\abs{x})}}{f(\kappa(x))\cdot\text{poly}(\abs{x})}$, then the FPQT algorithm outputs a multiplicative $\epsilon$-approximation to $g(x)$. This completes the proof.
\end{proof}

The following theorem is a folklore result in parameterized counting.
\begin{theorem}
    Fix $\epsilon>0$. Let $g$ be a function in \GapP, and $p\in\mathbb{N}[X]$ a polynomial such that, for any input $x$, $g(x)$ takes values in the range $[-2^{p(\abs{x})},\;2^{p(\abs{x})}]$. Further let $\kappa$ be a parameterization and $f:\mathbb{Z}^+\to\mathbb{Z}^+$ a computable function. Then there is an FPT algorithm that, for any input $x$, outputs an approximation to $g(x)$ to within an additive error of $\frac{\epsilon\cdot2^{p(\abs{x})}}{f(\kappa(x))\cdot\mathrm{poly}(\abs{x})}$. Furthermore, if $\kappa$ is a parameterization such that $\abs{g(x)}\geq\frac{2^{p(\abs{x})}}{f(\kappa(x))\cdot\mathrm{poly}(\abs{x})}$, then the FPT algorithm outputs a multiplicative $\epsilon$-approximation to $g(x)$.
\end{theorem}

\begin{proof}
    For any instance $x$ of $g$, we evaluate $m$ computational paths uniformly at random from the $2^{p(\abs{x})}$ possible. Let $\{X_i\}_{i=1}^m$ be the set of random variables such that $X_i$ takes the value $+1$ if the $i^\mathrm{th}$ computational path accepts instance $x$ and $-1$ otherwise. We then approximate $g(x)$ by $\tilde{g}(x)=\frac{2^{p(\abs{x})}}{m}\sum_{i=1}^mX_i$, which has expectation value $\mathbb{E}[\tilde{g}(x)]=g(x)$. By taking $m=f(\kappa(x))^2\cdot\text{poly}(\abs{x})\cdot\epsilon^{-2}$, it follows from the Chernoff-Hoeffding bound that
    \begin{equation}
        \textbf{Pr}\left[\abs{g(x)-\tilde{g}(x)} \geq \frac{\epsilon\cdot2^{p(\abs{x})}}{f(\kappa(x))\cdot\text{poly}(\abs{x})}\right] \leq e^{-\text{poly}(\abs{x})}. \notag
    \end{equation}
    Then, with high probability, we have
    \begin{equation}
        \abs{g(x)-\tilde{g}(x)} \leq \frac{\epsilon\cdot2^{p(\abs{x})}}{f(\kappa(x))\cdot\text{poly}(\abs{x})}. \notag
    \end{equation}
    Furthermore, if $\kappa$ is a parameterization such that $\abs{g(x)}\geq\frac{2^{p(\abs{x})}}{f(\kappa(x))\cdot\text{poly}(\abs{x})}$, then we have
    \begin{equation}
        \abs{g(x)-\tilde{g}(x)} \leq \epsilon\abs{g(x)}. \notag
    \end{equation}
    This completes the proof.
\end{proof}

The results of this section concern approximately counting the number of accepting paths when there are a large number of them. However, it is also possible to approximately count the number of accepting paths when there are a small number of them provided we have access to an oracle that solves the decision problem~\cite{meeks2019randomised, dell2020approximately}. We remark that our results also apply to combinatorial structures of fixed-parameter tractable size. In particular, we can approximate the evaluation of Jones polynomials of the plat closure of braids with a fixed-parameter tractable number of crossings.

\subsection{Quantum Merlin Arthur Proofs}
\label{section:QuantumMerlinArthurProofs}

In this section we study the complexity class $\QMA_k$ and its connection to quantum parameterized complexity. Recall that $\QMA_k$ is the complexity class consisting of all languages for there exists a \QMA verification procedure on $k$ witness qubits. Marriott and Watrous~\cite{marriott2005quantum} proved that $\QMA_{\log}=\BQP$, we follow their analysis to establish an \FPQT algorithm for any problem in $\QMA_k$ when parameterized by the witness length $k$.

\begin{theorem}
    Let $L\subseteq\Sigma^*$ be a language in $\QMA_k$ and let $\kappa:\Sigma^*\to\mathbb{Z}^+$ be the parameterization  with $\kappa(x)=k$ for all $x \in L$, then $(L,\kappa)$ is in \FPQT.
\end{theorem}

\begin{proof}
    Since $L$ is in $\QMA_k(\frac{2}{3},\frac{1}{3})$, by Ref.~\cite[Theorem 3.3]{marriott2005quantum}, we have that $L$ is in $\QMA_k(c,s)$ with $c=1-\frac{1}{3}2^{-k}$ and $s=\frac{1}{3}2^{-k}$. Let $\mathcal{A}$ be a $k$-qubit verification procedure for $L$ with $n$ workspace qubits and let $\{\Pi_0, \Pi_1\}$ be a measurement defined by the projectors
    \begin{equation}
        \Pi_0 \coloneqq \ketbra{0}{0}\otimes\mathbb{I}_{k+n-1} \quad\text{and}\quad \Pi_1 \coloneqq \ketbra{1}{1}\otimes\mathbb{I}_{k+l-1}, \notag
    \end{equation}
    which decides whether $\mathcal{A}$ accepts or rejects.
    For each $x\in\Sigma^*$, we define an operator $\mathcal{Q}_x$ by
    \begin{equation}
       \mathcal{Q}_x \coloneqq \left(\mathbb{I}_k\otimes\bra{0^n}\right)\mathcal{A}^\dagger(x)\Pi_1\mathcal{A}(x)\left(\mathbb{I}_k\otimes\ket{0^n}\right). \notag
    \end{equation}
    Since $\mathcal{Q}_x$ is positive semidefinite, then there exists a quantum state $\ket{\psi}$, such that, if $x \in L$, then
    \begin{equation}
        \Tr(\mathcal{Q}_x) \geq \expval{\mathcal{Q}_x}{\psi} = \textbf{Pr}[\text{$\mathcal{A}(x,\ket{\psi})$ accepts}] \geq 1-\frac{1}{3}2^{-k} \geq \frac{2}{3}. \notag
    \end{equation}
    Similarly, if $x \notin L$, then
    \begin{equation}
        \Tr(\mathcal{Q}_x) \leq 2^k\expval{\mathcal{Q}_x}{\psi} = 2^k\textbf{Pr}[\text{$\mathcal{A}(x,\ket{\psi})$ accepts}] \leq \frac{1}{3}. \notag
    \end{equation}
    We now establish an FPQT algorithm $\mathcal{B}$ for deciding $(L,\kappa)$. The algorithm $\mathcal{B}$ constructs the maximally mixed state $2^{-k}\mathbb{I}_k$ on $k$ qubits and then runs the verification procedure $\mathcal{A}$ with input $2^{-k}\mathbb{I}_k$. We then have
    \begin{equation}
         \textbf{Pr}[\text{$\mathcal{B}(x)$ accepts}] = 2^{-k}\Tr(\mathcal{Q}_x). \notag
    \end{equation}
    Hence, if $x \in L$, then
    \begin{equation}
        \textbf{Pr}[\text{$\mathcal{B}(x)$ accepts}] \geq \frac{2}{3}2^{-k}. \notag
    \end{equation}
    Otherwise, if $x \notin L$, then
    \begin{equation}
        \textbf{Pr}[\text{$\mathcal{B}(x)$ accepts}] \leq \frac{1}{3}2^{-k}. \notag
    \end{equation}
    Since these probabilities are bounded away from one another by an inverse exponential in $k$, by a simple gap amplification argument, we obtain an FPQT algorithm for $(L,\kappa)$. Hence, $(L,\kappa)$ is in \FPQT.
\end{proof}

\section*{Acknowledgements}

We thank M\'{a}ria Kieferov\'{a}, Kitty Meeks, Ashley Montanaro, Stephen Piddock, and Youming Qiao for helpful discussions. MJB was supported by the Australian Research Council (ARC) Centre of Excellence for Quantum Computation and Communication Technology (CQC2T), project number CE170100012, and in part by the National Science Foundation under Grant No. PHY-1748958 while visiting the KITP. ZJ was supported by the ARC Discovery Project DP200100950. RLM was supported by the QuantERA ERA-NET Cofund in Quantum Technologies implemented within the European Union's Horizon 2020 Programme (QuantAlgo project), EPSRC grants EP/L021005/1, EP/R043957/1, and EP/T001062/1, and the ARC Centre of Excellence for Quantum Computation and Communication Technology (CQC2T), project number CE170100012. MESM was supported by the ARC Centre of Excellence for Quantum Computation and Communication Technology (CQC2T), project number CE170100012 and a scholarship top-up and extension from the Sydney Quantum Academy. ATES was supported by an Australian Government Research Training Program Scholarship, the ARC Centre of Excellence for Quantum Computation and Communication Technology (CQC2T), project number CE170100012, and a scholarship top-up and extension from the Sydney Quantum Academy. No new data were created during this study.

\bibliography{bibliography}

\end{document}